\documentclass[conference]{IEEEtran}

\usepackage[cmex10]{amsmath}
\usepackage{amssymb}
\usepackage{xcolor}
\usepackage{pgf}
\usepackage{pgfplots}
\usetikzlibrary{chains,positioning,arrows,calc}
\usetikzlibrary{arrows,shapes,automata,backgrounds,petri,patterns}
\usetikzlibrary{external}

\def\X{{\mathbf X}}
\def\x{{\mathbf x}}
\newcommand\EXP[2]{\ensuremath{\underset{#1}{\mathbb{E}}\left[ #2 \right]}}
\newcommand{\reff}[1]{(\ref{#1})}
\renewcommand\Pr[1]{\text{Pr}\left[{#1}\right]}
\renewcommand\P[2]{\text{P}_{#1}\left({#2}\right)}
\newcommand\sgn[1]{\text{sgn}\left({#1}\right)}

\newtheorem{proposition}{Proposition}
\newtheorem{lemma}{Lemma}
\newtheorem{definition}{Definition}
\begin{document}

\title{Canalizing Boolean Functions Maximize\\ the Mutual Information}

\author{\IEEEauthorblockN{Johannes Georg Klotz, David Kracht, Martin Bossert, and Steffen Schober}
\IEEEauthorblockA{Institute of Communications Engineering\\
Ulm University\\
Albert-Einstein-Allee 43,
89073 Ulm, Germany\\
Email: \{johannes.klotz, david.kracht, martin.bossert, steffen.schober\}@uni-ulm.de}}

\maketitle

\begin{abstract}
The ability of information processing in biologically motivated Boolean networks is of interest in recent information theoretic research. One measure to quantify this ability is the well known mutual information. Using Fourier analysis we show that canalizing functions maximize the mutual information between an input variable and the outcome of the function. We proof our result for Boolean functions with uniform distributed as well as product distributed input variables. 
\end{abstract}

\IEEEpeerreviewmaketitle

\section{Introduction}
In systems and computational biology Boolean networks are widely used to model various dependencies. One major application of such networks are regulatory networks \cite{K69, CKR04}. In this context Boolean networks in general and consequently Boolean functions have been extensively studied in the past.

One focus of research is the information processing in such networks. It has been stated, that networks, whose dynamical behavior is critical, i.e., at the edge between stable and chaotic behavior, have somehow an optimized information processing ability \cite{L90, SML96, LF00}. Further, it has been shown, that biologically motivated networks have critical dynamics \cite{K69,K74}. The dynamics depend on two major characteristics of the network, the topology and the choice of functions. Here, so-called \emph{canalizing} Boolean functions seem to have a stabilizing effect \cite{KPS03,KPST04}. It has been shown \cite{SJ08,KFG11}, that regulatory networks, such as the regulatory network of \emph{Escherichia Coli}, consist mainly of canalizing functions. 

One way to formalize and to quantify information procession abilities is the well-known mutual information based on Shannon's theory \cite{S48}. The mutual information is a measure for the statistical relation between some input and output of a system, as for example shown by Fano's inequality \cite{HV94}. It has already been applied in the context of Boolean networks and functions, such as cellular automata \cite{L90}, random Boolean networks \cite{LF00,RKL08} and iterated function systems \cite{CY90}.

In this paper, we use Fourier analysis to investigate Boolean functions. This has first been addressed in \cite{G59}. Recently a relation between the mutual  information and the Fourier spectra has been derived \cite{HSB}. In particular, it has been shown that the mutual information between one input variable and the output only depends on two coefficients, namely the zero coefficient and the first order coefficient of this variable. Further, a relation between the canalizing property and these coefficients has been derived \cite{KSB11}. In this paper, we combine two approaches to show that canalizing functions maximize the mutual information for a given expectation value of the functions output.

The remainder of this paper is structured as follows: In the next section we will introduce some basic definitions and notation. In particular we will introduce the concepts of Fourier analysis as far as they are relevant to this work. In Section \ref{sec:MI} our main results are proven in two steps. First, we will address Boolean function with uniform distributed input variables, secondly, the result is extended to the more general product distributed case. This is followed by a short discussion of our finding (Section \ref{sec:discussion}), before we will conclude with some final remarks.

\section{Basic Definitions and Notation}
\subsection{Boolean Functions and Fourier Analysis}
A Boolean function (BF) $f \in \mathcal F_n = \{ f : \Omega^n  \rightarrow \Omega\}$ with $\Omega=\{-1,+1\}$ maps n-ary input tuples to a binary output. Let us consider $\x=(x_1, x_2, \ldots, x_n )$ as an instance of a product distributed random vector 
$\X=(X_1, X_2, \ldots, X_n )$, i.e., its probability density functions can be written as 
\begin{equation*}
\P{\X}{\x}=\Pr{\X=\x} = \prod_i \P{X_i}{x_i}.
\end{equation*}
Furthermore, let $\mu_i$ be the expected value of $X_i$, i.e., $\mu_i=\EXP{}{X_i}$ and let
\begin{equation} \label{eq:sigma}
\sigma_i=\sqrt{1-\mu_i^2}
\end{equation}
be the standard deviation of $X_i$.
It can be easily seen that 
\begin{equation*}
\P{X_i}{a_i} = \frac{1+a_i\cdot \mu_i}{2}, \text{~for~} a_i \in \{-1,+1\}.
\end{equation*}
It is well known that any BF $f$ can be expressed by the following sum, 
called Fourier-expansion,
\begin{equation*}
f(\x) = \sum_{U \subseteq [n]} \hat{f}(U)\cdot \phi_U(\x),
\end{equation*}
where $[n] = \{1,2,\ldots,n\}$ and
\begin{equation} \label{eq:phi}
\phi_U(\x)=\prod_{i \in U} \frac{x_i-\mu_i}{\sigma_i}.
\end{equation}
For $U=\emptyset$ we define  $\phi_{\emptyset}(\x) = 1$.
The Fourier coefficients $\hat{f}(U)$ can be recovered by
\begin{equation} \label{eq:fourier}
\hat{f}(U) = \sum_{\x \in  } \P{\X}{\x}\cdot f(\x)\cdot \phi_U(\x).
\end{equation}

If the input variables $X_i$ are uniformly distributed, i.e., 
$\mu_i=0$ and $\sigma_i=1$, Eq. \reff{eq:phi} reduces to
\begin{equation*} 
\chi_U(\x)\equiv\phi_U(\x)=\prod_{i \in U} x_i,
\end{equation*}
and consequently, as  $\P{\X}{\x} = 2^{-n}$  for all $\x$, Eq. \reff{eq:fourier} reduces to
\begin{equation*}
\hat{f}(U) = 2^{-n}\sum_{\x } f(\x)\cdot  \chi_U(\x).
\end{equation*}

\subsection{Canalizing Function}
\begin{definition}\label{def:canalizing}
A BF is called canalizing in variable $i$, if there exists a Boolean restrictive value $a_i \in \{-1,+1\}$ such that the function 
\begin{equation}
f(\x|_{x_i=a_i})=  b_i,
\end{equation}
for all $x_1, ... x_{i-1}, x_{i+1} .... x_n$, where $b_i \in \{-1,+1\}$ is a constant. 
\end{definition}

As shown in \cite{KSB11} the Fourier coefficients of canalizing functions then fulfill the following conditions:
\begin{equation}\label{eq:cond:1}
\hat{f}(\emptyset) +   \hat{f}(\{i\}) \cdot  \phi_{\{i\}}(a_i)  = b_i.
\end{equation}
Hence, as stated in \cite{KSB11}, a BF is canalizing in $i$, if and only if $\hat{f}(\emptyset)$ and $\hat{f}(\{i\})$ fulfill Eq.~\reff{eq:cond:1}.
Further, it can be easily seen, that in the uniform distributed case 
\begin{equation}\label{eq:b}
b_i=\sgn{\hat f(\emptyset)},
\end{equation}
where $\sgn{\cdot}$ gives the sign.

\section{Mutual Information of Boolean Functions}
\label{sec:MI}
The mutual information (MI) between two random variables is defined as:
$$
MI(Y;X)=H(Y)-H(Y|X),
$$
where
$$
H(X)=-\sum_{x\in \mathcal{X}}\P{X}{x}\log_2(\P{X}{x})
$$
is Shannon's entropy in bits of some discrete random variable $X$ with its domain $ \mathcal{X}$. For the special case that $|\mathcal{X}|=2$, it reduces to the binary entropy function:
$$
h(p)=-p\log_2(p)-(1-p)\log_2(1-p),
$$
with $p=\P{X}{+1}$.
Further, $H(Y|X)$ is the conditional entropy between two discrete random variables $X\in \mathcal{X}$ and $Y\in \mathcal{Y}$
\begin{align*}
H(Y|X)=\sum_{x\in \mathcal{X}}\P{X}{x}H(Y|X=x),
\end{align*}
with
$$
H(Y|X=x)=-\sum_{y\in \mathcal{Y}}\P{Y|X}{y|x}\log_2\P{Y|X}{y|x}.
$$

It has been shown in \cite{HSB}, that the mutual information between one input variable $i$ and the output of a boolean function is given as:
\begin{align}\label{eq:MIi}
&MI(f(\X);X_i) = h\left(\frac{1}{2}\left(1+\hat f(\emptyset)\right)\right)\\ \nonumber
&-\EXP{X_i}{h\left(\frac{1}{2}\left(1+\hat f(\emptyset) + \hat f(\{i\}) \phi_{\{i\}}(X_i)\right)\right)}.
\end{align}

The fact that the MI is only dependent on $\hat{f}(\emptyset)$ and $\hat{f}(\{i\})$ coincides with our statement in the previous section, that the canalizing property also depend on these two coefficients. Hence, we will only focus on those two Fourier coefficients in the following considerations.  The remaining coefficient can be chosen arbitrarily and have no influence on our findings. Also, the number of input variables $n$ does not restrict our investigations, it only determines the possible values $\hat{f}(\emptyset)$ and $\hat{f}(\{i\})$, since they are a multiple of $2^{-n}$.
\subsection{Mutual Information under Uniform Distribution}
For sake of clarity and ease of comprehension we will first focus on canalizing functions in the uniform distributed case. In the next section we will then generalize this result to product distributed variables.
\begin{proposition}\label{prop:uniform}
Let $f$ be a Boolean function with uniform distributed inputs. For a given and fixed $\hat f (\emptyset )$, the mutual information (see Eq. \reff{eq:MIi}) between one input variable $i$ and the output of $f$, is maximized by all functions, which are canalizing in variable $i$.
\end{proposition}
\begin{proof}
Since $\hat f (\emptyset )$ is constant, the only remaining degree of freedom in Eq. \reff{eq:MIi} is $\hat f( \{i\})$. First, we show that the mutual information is convex with respect to $\hat f(\{i\})$. Since the first summand of Eq. \reff{eq:MIi} only depends on $\hat f(\emptyset)$ we can consider it as constant. We hence can focus on the second part, which we can write as:
$$
\EXP{X_i}{-h\left(\frac{1}{2}\left(1+\hat f(\emptyset) + \hat f(\{i\}) \chi_{\{i\}}(X_i)\right)\right)}.
$$
The binary entropy function $h$ is concave, and since its argument is an affine mapping, $-h\left(\frac{1}{2}\left(1+\hat f(\emptyset) + \hat f(\{i\}) \chi_{\{i\}}(X_i)\right)\right)$ is convex \cite{Boyd04}.
Finally the expectation is a non-negative weighted sum, which preserves convexity, hence the mutual information is convex.

Obviously for  $\hat f(\{i\})=0$, the mutual information is minimized, hence, due to the convexity, the maximum can be found on the boundaries of the domain. The domain is limited by the non-negativity of the arguments of $h$, i.e.,
$$
0 \leq \frac{1}{2}\left(1+\hat f(\emptyset) + \hat f(\{i\}) \chi_{\{i\}}(X_i)\right) \leq 1.
$$
Hence, the boundaries are given by
$$
\hat f (\emptyset) + \hat f(\{i\}) \chi_{\{i\}}(X_i) = \pm 1.
$$
It can be seen from Definition \ref{def:canalizing}, that all functions, which are canalizing in variable $i$, are located on the boundary of the domain of the mutual information. These function are constrained with:
$$
\hat f(\{i\}) = \frac{b_i - \hat f (\emptyset)}{\chi_{\{i\}}(a_i)}.
$$
Since $a_i,b_i \in \{-1,+1\}$, there exists four such types of functions on the boundary.
Looking at their mutual information leads us to:
\begin{align*}
&MI(f(\X);X_i) = h\left(\frac{1}{2}\left(1+\hat f(\emptyset)\right)\right)\\ \nonumber
&-\EXP{X_i}{h\left(\frac{1}{2}\left(1+\hat f(\emptyset) +  \left(b_i - \hat f (\emptyset) \right)\frac{\chi_{\{i\}}(X_i)}{\chi_{\{i\}}(a_i)}\right)\right)},
\end{align*}
which yields in:
\begin{align*}
&MI(f(\X);X_i) = h\left(\frac{1}{2}\left(1+\hat f(\emptyset)\right)\right)\\ \nonumber
&-\P{X_i}{a_i}\underbrace{{h\left(\frac{1}{2}\left(1+\hat f(\emptyset) +  \left(b_i - \hat f (\emptyset) \right)\frac{\chi_{\{i\}}(a_i)}{\chi_{\{i\}}(a_i)}\right)\right)}}_{=0} \\ \nonumber
&-\P{X_i}{-a_i}{h\left(\frac{1}{2}\left(1+\hat f(\emptyset) +  \left(b_i - \hat f (\emptyset) \right)\frac{\chi_{\{i\}}(-{a}_i)}{\chi_{\{i\}}(a_i)}\right)\right)},
\end{align*}
and hence:
\begin{align} \label{eq:mi_can}
&MI(f(\X);X_i) = h\left(\frac{1}{2}\left(1+\hat f(\emptyset)\right)\right)\\ \nonumber
&-\P{X_i}{-{a}_i}{h\left(\frac{1}{2}\left(1+\hat f(\emptyset) +  \left(b_i - \hat f (\emptyset) \right)\frac{\chi_{\{i\}}(-{a}_i)}{\chi_{\{i\}}(a_i)}\right)\right)}.
\end{align}
For the uniform distributed case we write:
\begin{align*}
&MI(f(\X);X_i) = h\left(\frac{1}{2}\left(1+\hat f(\emptyset)\right)\right)\\ \nonumber
&-\frac{1}{2}h\left(\frac{1-b_i}{2}+ \hat f(\emptyset)\right).
\end{align*}
Due to $b_i=\sgn{\hat f(\emptyset)}$ and the symmetry of $h$, we finally get:
\begin{align*}
&MI(f(\X);X_i) = h\left(\frac{1}{2}\left(1+\hat f(\emptyset)\right)\right)-\frac{1}{2}h\left(|\hat f(\emptyset)|\right).
\end{align*}
Hence, the mutual information is independent from $a_i$ and $b_i$, which concludes the proof.
\end{proof}
\subsection{Mutual Information under Product Distribution}
The result from the previous section can be extended to product distributed input variables. We will see, that the probability distribution of the canalizing variable plays a key role in maximizing the MI.

\begin{proposition}
Let $f$ be a Boolean function with product distributed inputs. For a given and fixed $\hat f (\emptyset )$, the mutual information (see Eq. \reff{eq:MIi}) between one input variable $i$ and the output of $f$, is maximized by all functions, which are canalizing in variable $i$, where $a_i$ and $b_i$ are chosen as follows
\begin{align*}
(a_i,b_i) =
\begin{cases}
\left(\sgn{\mu_i},\sgn{\hat f(\emptyset)}\right) &  |\hat f (\emptyset) |\geq  |\mu_i|\\
\left(-\sgn{\mu_i},-\sgn{\hat f(\emptyset)}\right) &  |\hat f (\emptyset) |<  |\mu_i|
\end{cases}
\end{align*}
\end{proposition}
\begin{proof}
The first part of this proof goes along with the proof of Proposition \ref{prop:uniform}, where we simply replace $\chi_U(\x)$ by $\phi_U(\x)$. Hence, we can again show that the MI is convex and that the boundary consists of the canalizing functions. 
Starting from Eq. \reff{eq:mi_can}, we hence get
\begin{align*}
&MI(f(\X);X_i) = h\left(\frac{1}{2}\left(1+\hat f(\emptyset)\right)\right)- H(f(\X)|X_i),
\end{align*}
where
\begin{align*}
&H(f(\X)|X_i)\\ &= \P{X_i}{-{a}_i}{h\left(\frac{1}{2}\left(1+\hat f(\emptyset) -  \left(b_i - \hat f (\emptyset) \right) \frac{(\mu_i+a_i)^2}{\sigma_i^2}\right)\right)}.
\end{align*}
Obviously there are four possible sets of $a_i$ and $b_i$. However, the for each choice of $\hat f(\emptyset)$ there exists two possible choices of $a_i$ and $b_i$, of which only one maximizes the MI. 

Lets first look at the possible combinations of $a_i$ and $b_i$ in dependence of $\hat f(\emptyset)$.
From Parsevals theorem we know hat
$$
\hat f(\emptyset)^2 + \hat f({i})^2 \leq 1,
$$
and hence
$$
\hat f(\emptyset)^2 + \left( \frac{b_i - \hat f (\emptyset)}{\phi_{\{i\}}(a_i)}\right)^2 \leq 1,
$$
Solving that inequation for $\hat f(\emptyset)$ leads us to the possible sets of $a_i$ and $b_i$:
\begin{align*}
&a_i=\pm 1 \text{~and~} b_i=-1  &\text{~if~}& -1 \leq  \hat f(\emptyset) \leq -|\mu_i|\\
&a_i=-\sgn{\mu_i} \text{~and~} b_i=\pm 1 &\text{~if~}& -|\mu_i| \leq  \hat f(\emptyset) \leq |\mu_i| \\
&a_i= \pm 1 \text{~and~} b_i=1 &\text{~if~}& |\mu_i| \leq  \hat f(\emptyset) \leq 1 
%\mu_i \leq \hat f(\emptyset) &\leq 1 & \text{~if~} a_i=+1 \text{~and~} b_i=+1\\
%-\mu_i \leq \hat f(\emptyset) &\leq 1 & \text{~if~} a_i=-1 \text{~and~} b_i=+1\\
%-1 \leq \hat f(\emptyset) &\leq -\mu_i & \text{~if~} a_i=+1 \text{~and~} b_i=-1\\
%-1 \leq \hat f(\emptyset) &\leq \mu_i & \text{~if~} a_i=-1 \text{~and~} b_i=-1.
\end{align*}

Hence, to maximize the MI, we have to minimize $H(f(\X)|X_i)$ for each possible choice of $\hat f(\emptyset)$. We can rewrite $H(f(\X)|X_i)$ for all four combinations of $a_i$ and $b_i$ as follows:
\begin{align*}
s(\hat f(\emptyset))= \P{X_i}{-1}{h\left(\frac{\hat f(\emptyset) -\mu_i}{1-\mu_i}\right)}& \text{~if~} a_i=+1 \text{~and~} b_i=+1\\
t(\hat f(\emptyset))= \P{X_i}{+1}{h\left(\frac{\hat f(\emptyset) +\mu_i}{1+\mu_i}\right)}& \text{~if~} a_i=-1 \text{~and~} b_i=+1\\
q(\hat f(\emptyset))= \P{X_i}{-1}{h\left(\frac{\hat f(\emptyset) +1}{1-\mu_i}\right)} & \text{~if~} a_i=+1 \text{~and~} b_i=-1\\
r(\hat f(\emptyset))= \P{X_i}{+1}{h\left(\frac{\hat f(\emptyset) +1}{1+\mu_i}\right)} & \text{~if~} a_i=-1 \text{~and~} b_i=-1.
\end{align*}
Now, lets assume $-1 \leq  \hat f(\emptyset) \leq -|\mu_i|$, i.e., $b_i=-1$. Hence, have to compare $q(\hat f(\emptyset))$ and $r(\hat f(\emptyset))$ and search for the correct choice of $a_i$ which maximizes the MI. Lemma \ref{lem:a1} (can be found in the Appendix) shows that this is achieved if choosing $a_i=\sgn{\mu_i}$. 

If $|\mu_i| \leq  \hat f(\emptyset) \leq 1$, i.e., $b_i=+1$ we have again to compare the choices of $a_i=+1$ and $a_i=-1$. As above $a_i$ must be chosen to be $\sgn{\mu_i}$ in order to maximize the MI (see Lemma \ref{lem:a2} in the Appendix).

Now let $-|\mu_i| \leq  \hat f(\emptyset) \leq |\mu_i|$, hence we need to choose between $b_i=-1$ and $b_i=+1$. Lemma \ref{lem:a3} (Appendix) shows that in this case the MI is maximized if $b_i=-\sgn{\hat f(\emptyset))}$, which concludes the proof.
\end{proof}

\section{Discussion and Conclusion}
\label{sec:discussion}

\begin{figure}[tb]
  \centering
  \includegraphics{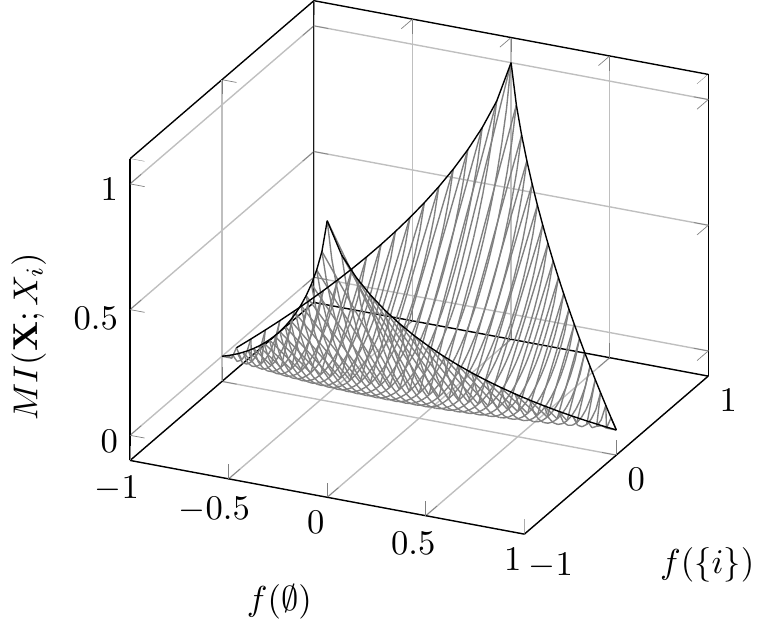}
%  
%  \begin{tikzpicture}
%    \begin{axis}[
%    xlabel=$\hat f(\emptyset)$, % Math-Mode ganz normal mit $...$
%    ylabel=$\hat f(\{i\})$, % Math-Mode ganz normal mit $...$
%    zlabel=$MI(\mathbf{X};X_i)$, % Math-Mode ganz normal mit $...$
%    xmin=-1.0, % Für's Finetuning, optional
%    xmax=1.0, % "
%    ymin=-1.0, % "
%    ymax=1.0, % "
%    legend pos=north west, % Position der Legende
%    legend style={cells={anchor=west}}, % Formatierung der Legende
%    width=212pt, % Absolute Breite des Plots
%    height=7cm, % Absolute Höhe des Plots
%    enlarge x limits=false, % true sieht komisch aus
%    enlarge y limits=false, % true sieht komisch aus
%    grid=major % Gitter für die Haupt-Ticks
%    % Achtung, letzte Option hat kein Komma!
%    ]
%	%psSurface[ngrid=.25 .25,incolor=yellow,fillcolor=blue, axesboxed] file {./graphs/3duniformnormal.dat};
%    \addplot3[surf,mesh, color=gray] file {./graphs/3duniformnormal.dat};
%    \addplot3[surf,mesh, color=gray] file {./graphs/3duniformnormal2.dat};
%
%    \addplot3[color = black] file {./graphs/3duniformcana1.dat};
%    \addplot3[color = black] file {./graphs/3duniformcana2.dat};
%    \addplot3[color = black] file {./graphs/3duniformcana3.dat};
%    \addplot3[color = black] file {./graphs/3duniformcana4.dat};
%
%
%    \end{axis}
%
%  \end{tikzpicture}
 \caption{Mutual information of BFs with uniform distributed input variables versus $\hat f(\emptyset)$ and $\hat f(\{i\})$, one can see that all in $i$ canalizing functions are located on the border (black line)}
  \label{fig:3d:uni}
\end{figure}
\begin{figure}[tb]
  \centering
  \includegraphics{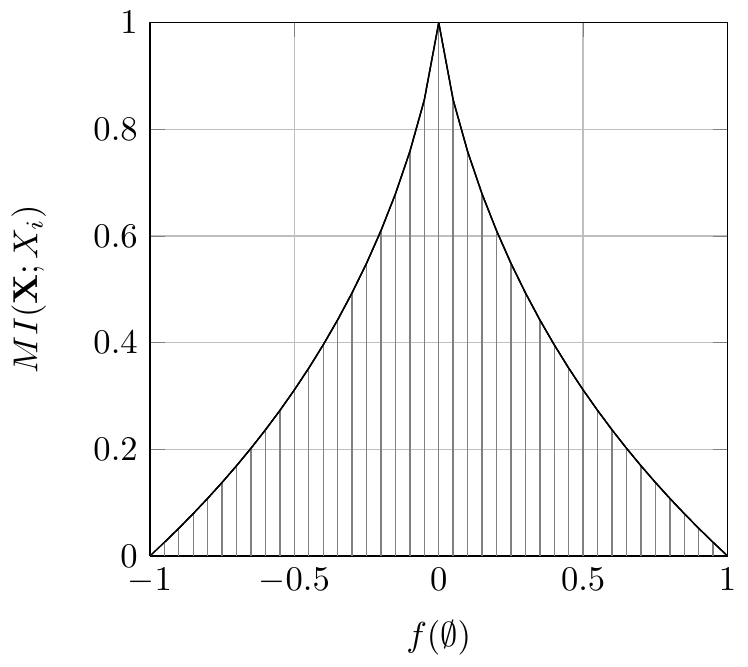}
  
%  \begin{tikzpicture}
%
%    \begin{axis}[
%    xlabel=$\hat f(\emptyset)$, % Math-Mode ganz normal mit $...$
%    ylabel=$MI(\mathbf{X};X_i)$, % Math-Mode ganz normal mit $...$
%    xmin=-1.0, % Für's Finetuning, optional
%    xmax=1.0, % "
%    ymin=0.0, % "
%    ymax=1.0, % "
%    legend pos=north west, % Position der Legende
%    legend style={cells={anchor=west}}, % Formatierung der Legende
%    width=212pt, % Absolute Breite des Plots
%    height=7cm, % Absolute Höhe des Plots
%    enlarge x limits=false, % true sieht komisch aus
%    enlarge y limits=false, % true sieht komisch aus
%    grid=major % Gitter für die Haupt-Ticks
%    % Achtung, letzte Option hat kein Komma!
%    ]
%
%    \addplot[mesh, color = gray] file {./graphs/2duniformnormal.dat};
%    \addplot[color = black] file {./graphs/2duniformcana1.dat};
%    \addplot[color = black] file {./graphs/2duniformcana2.dat};
%    \addplot[color = black] file {./graphs/2duniformcana3.dat};
%    \addplot[color = black] file {./graphs/2duniformcana4.dat};
%
%
%
%    \end{axis}
%
%  \end{tikzpicture}
  \caption{Mutual information of Fig. \ref{fig:3d:uni} projected in the ($\hat f(\emptyset),MI$)-plane, one can see that all in $i$ canalizing functions are located on the border (black line)}
  \label{fig:2d:uni}
\end{figure}

\begin{figure}[tb]
  \centering
      \includegraphics{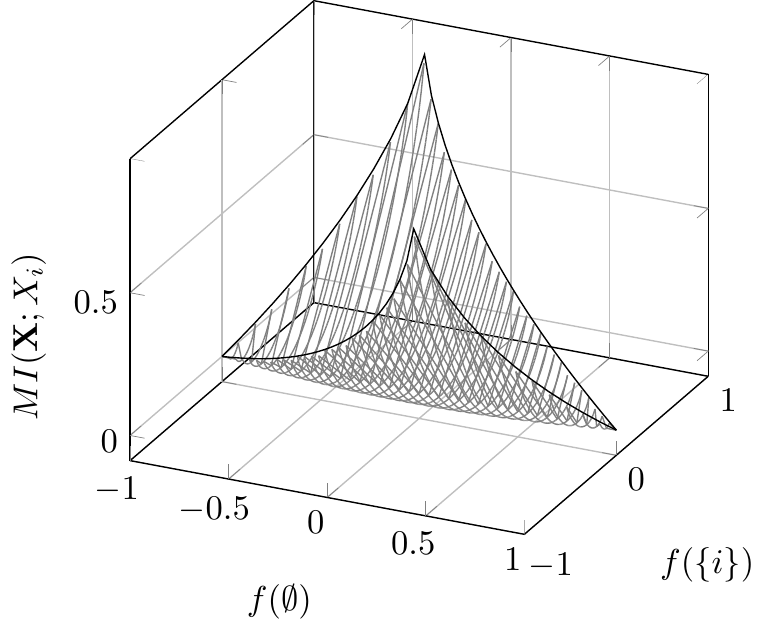}
  
%  \begin{tikzpicture}
%
%    \begin{axis}[
%    xlabel=$\hat f(\emptyset)$, % Math-Mode ganz normal mit $...$
%    ylabel=$\hat f(\{i\})$, % Math-Mode ganz normal mit $...$
%    zlabel=$MI(\mathbf{X};X_i)$, % Math-Mode ganz normal mit $...$
%    xmin=-1.0, % Für's Finetuning, optional
%    xmax=1.0, % "
%    ymin=-1.0, % "
%    ymax=1.0, % "
%    legend pos=north east, % Position der Legende
%    legend style={cells={anchor=west}}, % Formatierung der Legende
%    width=212pt, % Absolute Breite des Plots
%    height=7cm, % Absolute Höhe des Plots
%    enlarge x limits=false, % true sieht komisch aus
%    enlarge y limits=false, % true sieht komisch aus
%    grid=major % Gitter für die Haupt-Ticks
%    % Achtung, letzte Option hat kein Komma!
%    ]
%
%    \addplot3[surf,mesh, color=gray] file {./graphs/3dprodnormal.dat};
%    \addplot3[surf,mesh, color=gray] file {./graphs/3dprodnormal2.dat};
%
%    \addplot3[color = black] file {./graphs/3dprodcana1.dat};
%    %\addplot3[color = black] file {./graphs/3dprodcana2.dat};
%    \addplot3[color = black] file {./graphs/3dprodcana2.dat};
%%    \draw[color = white] {0,0};
%
%
%    \end{axis}
%
%  \end{tikzpicture}
  \caption{Mutual information of a BF with product distributed input variables versus $\hat f(\emptyset)$ and $\hat f(\{i\})$, $p_i=0.3$,  one can see that all in $i$ canalizing functions are located on the border (black line)}
  \label{fig:3d:prod}
\end{figure}
\begin{figure}[tb]
  \centering
      \includegraphics{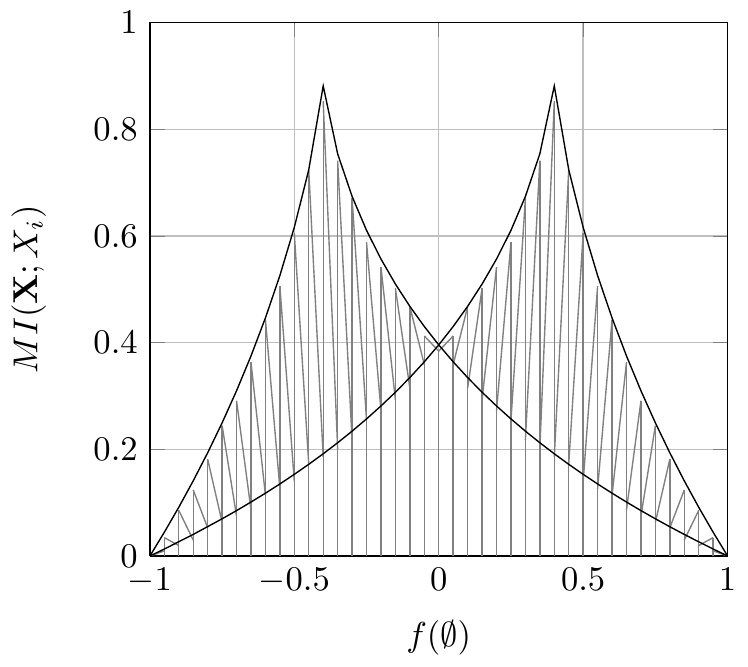}
  
%  \begin{tikzpicture}
%
%    \begin{axis}[
%    xlabel=$\hat f(\emptyset)$, % Math-Mode ganz normal mit $...$
%    ylabel=$MI(\mathbf{X};X_i)$, % Math-Mode ganz normal mit $...$
%    xmin=-1.0, % Für's Finetuning, optional
%    xmax=1.0, % "
%    ymin=0.0, % "
%    ymax=1.0, % "
%    legend pos=north west, % Position der Legende
%    legend style={cells={anchor=west}}, % Formatierung der Legende
%    width=212pt, % Absolute Breite des Plots
%    height=7cm, % Absolute Höhe des Plots
%    enlarge x limits=false, % true sieht komisch aus
%    enlarge y limits=false, % true sieht komisch aus
%    grid=major % Gitter für die Haupt-Ticks
%    % Achtung, letzte Option hat kein Komma!
%    ]
%
%    \addplot[mesh, color = gray] file {./graphs/2dprodnormal.dat};
%
%    \addplot[color = black] file {./graphs/2dprodcana1.dat};
%    \addplot[color = black] file {./graphs/2dprodcana2.dat};
%
%
%%    \addplot+[draw=blue,no marks,solid] file {./matlab/xy-data2.dat};
%%    \addlegendentry{\footnotesize Legende f. Plot, Math-Mode mit $...$};
%
%    %...
%
%    \end{axis}
%
%  \end{tikzpicture}
  \caption{Mutual information of Fig. \ref{fig:3d:prod} projected in the ($\hat f(\emptyset),MI$)-plane, one can see that all in $i$ canalizing functions are located on the border (black line)}
  \label{fig:2d:prod}
\end{figure}

To visualize our findings we plotted in Figure \ref{fig:3d:uni} a 3D diagram of the mutual information of a BF with uniform distributed input variables versus $\hat f(\emptyset)$ and $\hat f(\{i\})$. In Figure \ref{fig:2d:uni} we present a projection of the surface in the ($\hat f(\emptyset),MI$)-plane. It can be seen from these pictures that the canalizing function form the boundary of the domain of the MI. Further, the symmetry with respect to $a_i$ and $b_i=\sgn{\hat f(\emptyset)}$ can be seen.
In addition it becomes visible, that the mutual information also depends of the actual zero coefficient. This is mainly due to the first term of the MI (Eq. \ref{eq:MIi}), the entropy of the functions output.

In Figures \ref{fig:3d:prod} and \ref{fig:2d:prod} the same plots can be found for product distributed input variables, with $p_i=0.3$. Here, the skew of the mutual information towards the more probable canalizing value and the symmetry with respect to $b_i$ can bee seen. 

Our findings show the optimality of canalizing functions with respect to information processing abilities. Further, it has been stated in literature \cite{KPST04}, that canalizing functions have a stabilizing influence on the network dynamics. This supports the conjectures \cite{L90, SML96, LF00}, that these two properties are closely related. 

An open problem remains the impact of canalizing function on the mutual information between a set of variables and the function's output. One may presume that based on the results of this paper, that it is maximized by functions, which are somehow canalizing in all that variables.

\section*{Acknowledgment}
This work was supported by the German research council "Deutsche Forschungsgemeinschaft" (DFG) under Grant Bo 867/25-2. 

\appendix
\begin{lemma} \label{lem:a1}
If $-1 \leq  \hat f(\emptyset) \leq -|\mu_i|$, then
\begin{align*} 
q(\hat f(\emptyset))< r(\hat f(\emptyset)) &\text{~if~} \mu_i >0\\
q(\hat f(\emptyset))>r(\hat f(\emptyset))  &\text{~if~} \mu_i <0.
\end{align*}
\end{lemma}
\begin{proof}
%First, note that we can express $\frac{1}{2}+\frac{1}{2}g(+1,-1)$ as
%\begin{align*}
%\frac{1}{2}+\frac{1}{2}g(+1,-1) &=\frac{1}{2}+\frac{1}{2}\left(\hat f(\emptyset) -  \left(-1 - \hat f (\emptyset) \right) \frac{(\mu_i+1)^2}{1-\mu_i^2} \right) \\
%&=\frac{1}{2} \left(1+\hat f(\emptyset) \left( 1 +  \frac{1+\mu_i}{1-\mu_i}\right) +\frac{1+\mu_i}{1-\mu_i} \right) \\
%&=\frac{1}{2}\left(\hat f(\emptyset) \frac{2}{1-\mu_i} +\frac{1+\mu_i+1-\mu_i}{1-\mu_i} \right)\\
%&=\frac{1}{2}\left(\hat f(\emptyset) \frac{2}{1-\mu_i} +\frac{2}{1-\mu_i} \right) \\
%&=\frac{\hat f(\emptyset) +1}{1-\mu_i}
%\end{align*}
%and $\frac{1}{2}+\frac{1}{2}g(-1,-1)$ as
%\begin{align*}
%\frac{1}{2}+\frac{1}{2}g(-1,-1) &=\frac{1}{2}+\frac{1}{2}\left(\hat f(\emptyset) -  \left(-1 - \hat f (\emptyset) \right) \frac{(\mu_i-1)^2}{1-\mu_i^2}\right) \\
%&=\frac{1}{2}\left(1+ \hat f(\emptyset) \left( 1 +  \frac{1-\mu_i}{1+\mu_i}\right) +\frac{1-\mu_i}{1+\mu_i}\right) \\
%&=\frac{1}{2}\left(\hat f(\emptyset) \frac{2}{1+\mu_i} +\frac{1+\mu_i + 1-\mu_i}{1+\mu_i}\right)\\
%&=\frac{1}{2}\left(\hat f(\emptyset) \frac{2}{1+\mu_i} +\frac{2}{1+\mu_i}\right) \\
%&=\frac{\hat f(\emptyset) +1}{1+\mu_i},
%\end{align*}
First, lets recall, that
$$
q(\hat f(\emptyset))= \frac{1-\mu_i}{2}h\left(\frac{\hat f(\emptyset) +1}{1-\mu_i}\right)
$$
and
$$
r(\hat f(\emptyset))= \frac{1+\mu_i}{2}h\left(\frac{\hat f(\emptyset) +1}{1+\mu_i}\right)
$$
Lets assume that $\mu_i >0$. Due to the concavity and the parabolic form of $q(\hat f(\emptyset))$ and $r(\hat f(\emptyset))$, they can intersect at most two times. Obviously, $q(-1)=r(-1)=0$, and $q(-\mu_i)=0 < r(-\mu_i)$. Hence, if the slope of $r$ at $\hat f(\emptyset) =-1$ is larger than the slope of $s$, then $q(\hat f(\emptyset))< r(\hat f(\emptyset))$ on the interval $-1 \leq  \hat f(\emptyset) \leq -\mu_i$. 

Building the derivative of $q(\hat f(\emptyset))$ leads us to 
\begin{align*}
q'(\hat f(\emptyset))&= \P{X_i}{-1} \log\left(\frac{1-\frac{\hat f(\emptyset) +1}{1-\mu_i}}{\frac{\hat f(\emptyset) +1}{1-\mu_i}}\right)\\
&= \P{X_i}{-1} \log\left(\frac{-\mu_i-\hat f(\emptyset)}{\hat f(\emptyset) +1}\right)
\end{align*}
and similar
\begin{align*}
r'(\hat f(\emptyset))&= \P{X_i}{+1} \log\left(\frac{+\mu_i-\hat f(\emptyset)}{\hat f(\emptyset) +1}\right)
\end{align*}
One can see that
\begin{align*}
\lim_{f \to -1} \left(r'(\hat f(\emptyset))-q'(\hat f(\emptyset)) \right) = +\infty,
\end{align*}
which concludes the proof for $\mu_i >0$.
The proof for $\mu_i<0$ goes along the lines as for  $\mu_i>0$.
\end{proof}

\begin{lemma} \label{lem:a2}
If $|\mu_i| \leq \hat f(\emptyset) \leq 1$, then
\begin{align*} 
s(\hat f(\emptyset))< t(\hat f(\emptyset)) &\text{~if~} \mu_i >0\\
s(\hat f(\emptyset))>t(\hat f(\emptyset))  &\text{~if~} \mu_i <0.
\end{align*}
\end{lemma}
\begin{proof}
%Similar to the previous Lemma, we can write
%\begin{align*}
%\frac{1}{2}+\frac{1}{2}g(+1,+1) &=\frac{1}{2}+\frac{1}{2}\left(\hat f(\emptyset) -  \left(1 - \hat f (\emptyset) \right) \frac{(\mu_i+1)^2}{1-\mu_i^2} \right) \\
%&=\frac{1}{2} \left(1+\hat f(\emptyset) \left( 1 +  \frac{1+\mu_i}{1-\mu_i}\right) -\frac{1+\mu_i}{1-\mu_i} \right) \\
%&=\frac{1}{2}\left(\hat f(\emptyset) \frac{2}{1-\mu_i} +\frac{1-\mu_i-1-\mu_i}{1-\mu_i} \right)\\
%&=\frac{1}{2}\left(\hat f(\emptyset) \frac{2}{1-\mu_i} +\frac{-2\mu_i}{1-\mu_i} \right) \\
%&=\frac{\hat f(\emptyset) -\mu_i}{1-\mu_i}
%\end{align*}
%and 
%\begin{align*}
%\frac{1}{2}+\frac{1}{2}g(-1,+1) &=\frac{1}{2}+\frac{1}{2}\left(\hat f(\emptyset) -  \left(1 - \hat f (\emptyset) \right) \frac{(\mu_i-1)^2}{1-\mu_i^2}\right) \\
%&=\frac{1}{2}\left(1+ \hat f(\emptyset) \left( 1 +  \frac{1-\mu_i}{1+\mu_i}\right) -\frac{1-\mu_i}{1+\mu_i}\right) \\
%&=\frac{1}{2}\left(\hat f(\emptyset) \frac{2}{1+\mu_i} +\frac{1+\mu_i - 1+\mu_i}{1+\mu_i}\right)\\
%&=\frac{1}{2}\left(\hat f(\emptyset) \frac{2}{1+\mu_i} +\frac{2\mu_i}{1+\mu_i}\right) \\
%&=\frac{\hat f(\emptyset) +\mu_i}{1+\mu_i},
%\end{align*}
First, lets recall that
$$
s(\hat f(\emptyset))= \frac{1-\mu_i}{2}{h\left(\frac{\hat f(\emptyset) -\mu_i}{1-\mu_i}\right)}
$$
and
$$
t(\hat f(\emptyset))= \frac{1+\mu_i}{2}{h\left(\frac{\hat f(\emptyset) +\mu_i}{1+\mu_i}\right)}.
$$

Now we assume, that $\mu_i >0$. Due to the concavity and the parabolic form of  $s(\hat f(\emptyset))$ and $t(\hat f(\emptyset))$, they can intersect at most two times. Obviously, $s(+1)=t(+1)=0$, and $s(\mu_i)=0 < t(\mu_i)$. Hence, if the slope of $t$ at $\hat f(\emptyset) =-1$ is larger than the slope of $s$, then $s(\hat f(\emptyset))< t(\hat f(\emptyset))$ on the interval $\mu_i \leq  \hat f(\emptyset) \leq 1$. 

Building the derivative of $s(\hat f(\emptyset))$ leads us to 
\begin{align*}
s'(\hat f(\emptyset))&= \P{X_i}{-1} \log\left(\frac{1-\frac{\hat f(\emptyset) -\mu_i}{1-\mu_i}}{\frac{\hat f(\emptyset) -\mu_i}{1-\mu_i}}\right)\\
&= \P{X_i}{-1} \log\left(\frac{1-\hat f(\emptyset)}{\hat f(\emptyset) -\mu_i}\right)
\end{align*}
and similar
\begin{align*}
t'(\hat f(\emptyset))&= \P{X_i}{+1} \log\left(\frac{1-\hat f(\emptyset)}{\hat f(\emptyset) +\mu_i}\right)
\end{align*}
One can see that
\begin{align*}
\lim_{\hat f(\emptyset) \to 1} \left(t'(\hat f(\emptyset))-s'(\hat f(\emptyset)) \right) = -\infty,
\end{align*}
which concludes the proof for $\mu_i >0$.
The proof for $\mu_i<0$ goes along the lines as for  $\mu_i>0$.
\end{proof}

\begin{lemma} \label{lem:a3}
If $\mu_i > 0$ and $|\hat f(\emptyset)| \leq |\mu_i|$,  then 
\begin{align*} 
t(\hat f(\emptyset))< r(\hat f(\emptyset)) &\text{~if~} \hat f(\emptyset)<0\\
t(\hat f(\emptyset))> r(\hat f(\emptyset))  &\text{~if~} \hat f(\emptyset)>0.
\end{align*}
If $\mu_i < 0$ and $|\hat f(\emptyset)| \leq |\mu_i|$, then
\begin{align*} 
s(\hat f(\emptyset))< q(\hat f(\emptyset)) &\text{~if~} \hat f(\emptyset)<0\\
s(\hat f(\emptyset))> q(\hat f(\emptyset))  &\text{~if~} \hat f(\emptyset)>0.
\end{align*}
\end{lemma}
\begin{proof}
Lets first assume $\mu_i > 0$. One can easily see, that
\begin{align*}
&t(-\mu_i) = \P{X_i}{+1}{h\left(\frac{-\mu_i+\mu_i}{1+\mu_i}\right)}=0 \\
&< r(-\mu_i) = \P{X_i}{+1}{h\left(\frac{-\mu_i+1}{1+\mu_i}\right)},
\end{align*}
and
\begin{align*}
&t(\mu_i) = \P{X_i}{+1}{h\left(\frac{\mu_i+\mu_i}{1+\mu_i}\right)} \\
&> r(\mu_i) = \P{X_i}{+1}{h\left(\frac{\mu_i+1}{1+\mu_i}\right)}=0.
\end{align*}
Further,
\begin{align*}
t(0) &= \P{X_i}{+1}{h\left(\frac{\mu_i}{1+\mu_i}\right)}\\
&= \P{X_i}{+1}{h\left(1-\frac{\mu_i}{1+\mu_i}\right)} \\
&= \P{X_i}{+1}{h\left(\frac{1}{1+\mu_i}\right)} = r(0),
\end{align*}
which due to concavity of $t$ and $r$ proofs the first part of the Lemma. The proof of the second part goes along the lines.
\end{proof}



\bibliographystyle{IEEEtran}
\bibliography{literatur}

\end{document}